\RequirePackage{tikz}
\documentclass{article}

\usepackage{url}
\usepackage{comment}
\usetikzlibrary{automata,positioning}
\usepackage{subcaption}
\usepackage{amsmath,amsthm,amsfonts}
\usepackage[nointegrals]{wasysym}
\usepackage[margin=1.3in]{geometry}

\newtheorem{theorem}{Theorem}
\newtheorem{lemma}[theorem]{Lemma}
\newtheorem{corollary}[theorem]{Corollary}
\newtheorem{proposition}[theorem]{Proposition}%
\newtheorem{example}{Example}%
\newtheorem{remark}{Remark}%

\newtheorem{definition}{Definition}%

\newcommand{\ZZ}{\mathbb{Z}}
\newcommand{\End}{\mathrm{End}}
\newcommand{\ID}{\mathrm{id}}
\newcommand{\N}{\mathbb{N}}
\newcommand{\lang}{\mathcal{L}}
\newcommand{\INF}{{^\omega}}
\newcommand\xqed[1]{%
  \leavevmode\unskip\penalty9999 \hbox{}\nobreak\hfill
  \quad\hbox{#1}}
\newcommand\qee{\xqed{$\fullmoon$}}
\newcommand\Var{\mathrm{Var}}
\newcommand\Prob{\mathrm{Pr}}
\newcommand{\vl}{{\mid}}

\begin{document}

\title{On von Neumann regularity of cellular automata}
\author{Ville Salo}

\maketitle

\abstract{We show that a cellular automaton on a one-dimensional two-sided mixing subshift of finite type is a von Neumann regular element in the semigroup of cellular automata if and only if it is split epic onto its image in the category of sofic shifts and block maps. It follows from previous joint work of the author and T\"orm\"a that von Neumann regularity is a decidable condition, and we decide it for all elementary CA, obtaining the optimal radii for weak generalized inverses. Two sufficient conditions for non-regularity are having a proper sofic image or having a point in the image with no preimage of the same period. We show that the non-regular ECA $9$ and $28$ cannot be proven non-regular using these methods. We also show that a random cellular automaton is non-regular with high probability.}

\section{Introduction}

The von Neumann regular elements -- elements $a$ having a weak inverse $b$ such that $aba = a$ -- of cellular automaton (CA) semigroups are studied in \cite{CaGa20}. We show that in the context of cellular automata on one-dimensional two-sided mixing subshifts of finite type, von Neumann regularity coincides with the notion of split epicness onto the image, another generalized invertibility notion from category theory.

Question 1 of \cite{CaGa20} asks which of the so-called elementary cellular automata (ECA) are von Neumann regular. They determine this for all ECA except ones equivalent to those with numbers 6, 7, 9, 23, 27, 28, 33, 41, 57, 58 and 77, see the next section for the definition of the numbering scheme. 

What makes this question interesting is that von Neumann regularity of one-dimensional cellular automata is not obviously\footnote{Specifically, many things about ``one-step behavior'' of cellular automata (like surjectivity and injectivity) are decidable using automata theory, or the decidability of the MSO logic of the natural numbers under successor. No decision algorithm for split epicness using these methods is known. Split epicness is a first-order property where we quantify over cellular automata/block maps, and many such properties, for example conjugacy \cite{JaKa20}, are undecidable.} decidable -- clearly checking if a given CA $f$ has a weak inverse is semidecidable, but it is not immediately clear how to semidecide the nonexistence of a weak inverse. However, split epicness has been studied previously in \cite{SaTo15a}, and in particular it was shown there that split epicness of a morphism between two sofic shifts is a decidable condition. This means Question 1 of \cite{CaGa20} can in theory be decided algorithmically.

As the actual bound stated in \cite{SaTo15a} is beyond astronomical, it is an interesting question whether the method succeeds in actually deciding each case. Using this method we give a human proof that 9, 27, 28, 41 and 58 are not von Neumann regular, and we prove by computer that ECA 6, 7, 23, 33, 57 and 77 are von Neumann regular, answering the remaining cases of Question 1 of \cite{CaGa20}.

The von Neumann regular CA on this list have weak inverses of radius at most four. Non-regularity is proved in each case by looking at eventually periodic points of eventual period one and applying the method of \cite{SaTo15a}, namely the strong periodic point condition. For ECA 27, 41 and 58, non-regularity can also be proved by simply observing that their images are proper sofic. For ECA 9 and ECA 28, this method cannot be used, as the image is an SFT. The ECA $9, 27, 28, 58$ admit preimages of the same period for all periodic points, so the method used in \cite{CaGa20} cannot be used to prove their non-regularity. In particular $9$ and $28$ seem to require using the strong periodic point condition.

This is an extended version of the paper \cite{Sa21a}. The new results about ECA are the precise optimal radii for the weak inverses (one of those reported in \cite{Sa21a} was suboptimal, and we had only checked the optimality in one case), and the observation that $9$ and $28$ admit preimages of the same period for periodic points. We also prove that a random CA is not von Neumann regular with high probability (Theorem~\ref{thm:RandomCA}), which refutes an almost-conjecture stated in \cite{Sa21a}.

\section{Preliminaries}

The \emph{full shift} is $\Sigma^\ZZ$ where $\Sigma$ is a finite alphabet, carrying the product topology. It is homeomorphic to the Cantor set, thus compact and metrizable. It is a dynamical system under the \emph{shift} $\sigma(x)_i = x_{i+1}$. Its subsystems (closed shift-invariant subsets) are called \emph{subshifts}, and the bi-infinite words $x \in \Sigma^\ZZ$ are their \emph{points}.

A \emph{cellular automaton (CA)} is a shift-commuting continuous function $f : X \to X$ on a subshift $X$. The cellular automata on a subshift $X$ form a monoid $\End(X)$ under function composition. A CA $f$ is \emph{reversible} if $\exists g: f \circ g = g \circ f = \ID$ where $g$ is a cellular automaton. Reversibility is equivalent to bijectivity by a compactness argument.

A CA has a local rule, that is, there exists a \emph{radius} $r \in \N$ such that $f(x)_i$ is determined by $x \vl_{[i-r,i+r]}$ for all $x \in X$ (and does not depend on $i$). More generally, a \emph{neighborhood} is a finite subset $N \subset \ZZ$ such that $f(x)_i$ is determined by $x \vl_{i+N}$. The \emph{elementary cellular automata} (ECA) are the CA on the binary full shift $\{0,1\}^\ZZ$ which can be defined with radius $1$. There is a numbering scheme for such CA: If $n \in [0,255]$ has base $2$ representation $b_7b_6...b_1b_0$, then ECA number $n$ is the one mapping $f(x)_i = b_{(x_{[i-1,i+1]})_2}$ where $(x_{[i-1,i+1]})_2$ is the number represented by $x_{[i-1,i+1]}$ in base $2$.
This numbering scheme is from \cite{Wo83}. The usage of base 10 for $n$ in this notation is standard, and some CA researchers remember ECA by these numbers. From radius $2$ onward, we switch to hexadecimal notation.

We recall \cite[Definition~3]{CaGa20}: define maps $R, S : \{0,1\}^\ZZ \to \{0,1\}^\ZZ$ by the formulas $R(x)_i = x_{-i}$ (this reverses the configuration) and $S(x)_i = 1-x_i$ (this flips all the bits in the configuration). Two cellular automata $f, g \in \End(\{0,1\}^\ZZ)$ are \emph{equivalent} if $f \in \langle S \rangle \circ g \circ \langle S \rangle \cup \langle S \rangle \circ R \circ g \circ R \circ \langle S \rangle$, where $\circ$ denotes function composition and $\langle S \rangle = \{\mbox{id}, S\}$. This means that we may pre- or postcompose by the bit flip $S$, and we may conjugate the CA by the reversal $R$. Composing by $R$ from just one side almost never gives a cellular automaton, so this is not allowed.

A subshift can be defined by forbidding a set of finite words from appearing as subwords of its points (which themselves are infinite words), and this is in fact a characterization of subshifts. A subshift is \emph{of finite type} or \emph{SFT} if it can be defined by a finite set of forbidden words, and \emph{sofic} if it can be defined by forbidding a regular language, in the sense of automata and formal languages. A \emph{proper sofic} subshift is one that is sofic but not SFT.

The \emph{language} $\lang(X)$ of a subshift $X$ is the set of finite words that appear in its points. A subshift $X$ is \emph{mixing} if for all words $u, v$ appearing in the language of $X$, for all large enough $n$ some word $uwv$ with $\vl w \vl = n$ appears in the language of $X$. In the case of a sofic shift $X$ with language $L$, we can simplify this to $\exists m: \forall u,v \in L: \exists w \in L: \vl w \vl = m \wedge uwv \in L$.

The language of a subshift determines it uniquely, and thus a convenient way to define a subshift is to describe its language. For a language $L$ which is \emph{extendable} meaning $\forall v \in L: \exists u, w \in L: uvw \in L$, there exists a unique smallest subshift whose language contains $L$. We write this subshift as $\lang^{-1}(L)$. For $L$ regular, this gives a sofic shift. We note that this is not an inverse for $\lang$, but rather $\lang(\lang^{-1}(L))$ is the \emph{factor closure} of $L$, i.e.\ $v \in \lang(\lang^{-1}(L)) \iff \exists u, w: uvw \in L$.

If $u \in A^*$ is a finite nonempty word, we write $u^\ZZ$ for the $\vl u\vl$-periodic point (i.e.\ fixed point of $\sigma^{\vl u\vl }$) in $A^\ZZ$ whose subword at $[0, \vl u\vl -1]$ is equal to $u$. We say two points are \emph{left-asymptotic} if their values agree at $-n$ for large enough $n$, and symmetrically define \emph{right-asymptoticity}. Points are \emph{asymptotic} if they are left- and right-asymptotic. The notation $\INF uwv \INF$ denotes a point left-asymptotic to (a shift-image of) $u^\ZZ$ and right-asymptotic to (a shift-image of) $v^\ZZ$, with the word $w$ between. When the positioning is important, we include a decimal point.

The \emph{de Bruijn graph} (with parameter $n$) is the graph with nodes for words of length $n-1$. When considering a cellular automaton with neighborhood consisting of $n$ consecutive positions, we typically label the transition from $au$ to $ub$ ($|u| = n-2$) by the image of $aub$ under the local rule. This is a convenient way to write the local rule of an ECA (in the case $n = 3$). Images of configurations can be found by following the corresponding path in the de Bruijn path (looking at consecutive pairs of symbols) and reading the edge labels.

In some parts of the paper, we assume knowledge of automata theory and formal languages. In automata theoretic terms, the de Bruijn graph can be directly seen as a non-deterministic finite-state automaton for the language of a CA with neighborhood $[0, n-1]$ (note that shifting the neighborhood does not affect the image): simply consider every node as both initial and final. This is called the \emph{de Bruijn automaton}.

See standard references for more information on symbolic dynamics \cite{LiMa95} or automata theory and formal languages \cite{HoMoUl06}.

A \emph{semigroup} is a set with an associative product, denoted by juxtaposition. A \emph{category} is specified by a class of objects, for each pair of objects $A, B$ a class of morphisms $f : A \to B$, a composition $gf = g \circ f : A \to C$ for morphisms $f : A \to B, g : B \to C$, and identity morphisms $\mathrm{id} : A \to A$ for each object. The composition should be associative, and the identity morphism should act trivially in composition. See \cite{Ma71} for more information on categories (very little is needed here).

\section{Split epicness and von Neumann regularity}

In this section, we show split epicness and von Neumann regularity are equivalent concepts on mixing SFTs. On full shifts, this is simply a matter of defining these terms.

If $S$ is a semigroup, then $a \in S$ is \emph{(von Neumann) regular} if $\exists b \in S: a b a = a \wedge b a b = b$. We say $b$ is a \emph{generalized inverse} of $a$. If $aba = a$ (but not necessarily $bab = b$), then $b$ is a \emph{weak inverse} of $a$. (More properly, one might call this a weak generalized inverse, but we use the shorter term as it is unambiguous.)

\begin{lemma}
If $a$ has a weak inverse, then it has a generalized inverse and thus is regular.
\end{lemma}

\begin{proof}
If $aba = a$, then letting $c = bab$, we have
$aca = ababa = aba = a$
and
$cac = bababab = babab = bab = c$. 
\end{proof}

If $\mathcal{C}$ is a category, a morphism $f : X \to Y$ is \emph{split epic} if there is a morphism $g : Y \to X$ such that $f \circ g = \ID_Y$. Such a $g$ is called a \emph{right inverse} or a \emph{section}.

Note that, in general, category-theoretic concepts depend on the particular category at hand, but if $\mathcal{C}$ is a full subcategory of $\mathcal{D}$ (meaning a subcategory induced by a subclass of the objects, by taking all the morphisms between them), then split epicness for a morphism $f : X \to Y$ where $X, Y$ are objects of $\mathcal{C}$ means the same in both.

We are in particular interested in the categories $K2, K3, K4$ (in the naming scheme of \cite{SaTo15a}) with respectively SFTs, sofic shifts, or all subshifts as objects, and \emph{block maps}, i.e. shift-commuting continuous functions $f : X \to Y$ as morphisms. Note that $K2$ is a full subcategory of $K3$, which in turn is a full subcategory of $K4$. By the previous paragraph, the choice does not really matter.

The following theorem is essentially only a matter of translating terminology, and works in many concrete categories.

\begin{theorem}
\label{thm:Equivalence}
Let $X$ be a subshift, and $f : X \to X$ a CA. Then the following are equivalent:
\begin{itemize}
\item $f : X \to f(X)$ has a right inverse $g : f(X) \to X$ which can be extended to a morphism $h : X \to X$ such that $h\vl_{f(X)} = g$,
\item $f$ is regular as an element of $\End(X)$.
\end{itemize}
\end{theorem}

\begin{proof}
Suppose first that $f$ is regular, and $h \in \End(X)$ satisfies $fhf = f$ and $hfh = h$. Then the restriction $g = h\vl_{f(X)} : f(X) \to X$ is still shift-commuting and continuous, and $\forall x: fg(f(x)) = f(x)$ implies that for all $y \in f(X)$, $fg(y) = y$, i.e. $g$ is a right inverse for the codomain restriction $f : X \to f(X)$ and it extends to the map $h : X \to X$ by definition.

Suppose then that $fg = \ID_{f(X)}$ for some $g : f(X) \to X$, as a right inverse of the codomain restriction $f : X \to f(X)$. Let $h : X \to X$ be such that $h\vl_{f(X)} = g$, which exists by assumption. Then $fh(f(x)) = fg(f(x)) = f(x)$. Thus $f$ is regular, and $hfh$ is a generalized inverse for it by the proof of the previous lemma. 
\end{proof}

Note that when $X$ is a full shift, extending morphisms is trivial: if $g : Y \to X$ has been defined by a local rule of radius $r$, defined on the words of length $2r+1$ in $Y$, we can simply fill in the local rule arbitrarily to obtain a cellular automaton $h : X \to X$ with $h\vl_Y = g$. This does not work for general $Y, X$, as the image of $h$ might not be contained in $X$, but it turns out that partially defined CA on mixing SFTs can be extended in an analogous way.

For two subshifts $X, Y$ write $X \overset{\mathrm{per}}{\rightarrow} Y$ if the period of every periodic point of $X$ is divisible by the period of some periodic point of $Y$. The following is Boyle's extension lemma \cite{Bo83}.

\begin{lemma}
Let $\tilde{S} \subset S$ be subshifts, let $T$ be a mixing SFT. Suppose $\tilde{\phi} : \tilde{S} \to T$ is shift-commuting and continuous, and $S \overset{\mathrm{per}}{\rightarrow} T$. Then there exists a shift-commuting continuous map $\phi : S \to Y$ such that $\phi\vl_{\tilde S} = \tilde{\phi}$.
\end{lemma}

\begin{lemma}
Let $X$ be a mixing SFT, $Y \subset X$ a subshift, and $g : Y \to X$ any shift-commuting continuous map. Then $g = h\vl_Y$ for some cellular automaton $h : X \to X$.
\end{lemma}

\begin{proof}
In the previous lemma, let $S = T = X, \tilde S = Y$, $\tilde{\phi} = g$. The condition $S \overset{\mathrm{per}}{\rightarrow} T$, i.e.\ $X \overset{\mathrm{per}}{\rightarrow} X$, is trivial.
\end{proof}

\begin{theorem}
\label{thm:MixingSFTEquivalence}
Let $X$ be a mixing SFT, and $f : X \to X$ a CA. Then the following are equivalent for $n = 2,3,4$:
\begin{itemize}
\item $f : X \to f(X)$ is split epic in $Kn$.
\item $f$ is regular as an element of $\End(X)$.
\end{itemize}
\end{theorem}

\begin{proof}
By Theorem~\ref{thm:Equivalence}, we need to check that ``$f : X \to f(X)$ has a right inverse $g : f(X) \to X$ which can be extended to a morphism $h : X \to X$ such that $h\vl_{f(X)} = g$'' is equivalent to split epicness. The forward implication is trivial, and the backward implication follows from the previous lemma, since every $g : f(X) \to X$ extends to such $h$.
\end{proof}


Theorem~\ref{thm:MixingSFTEquivalence} clearly implies that regularity respects equivalence of CA (this is not difficult to obtain directly from the definition either).

\begin{corollary}
If $f, g \in \End(\{0,1\}^\ZZ)$ are equivalent, then $f$ is regular if and only if $g$ is regular.
\end{corollary}

\section{Deciding split epicness}

We recall the characterization of split epicness \cite[Theorem~1]{SaTo15a}. This is Theorem~\ref{thm:SplitEpicSolved} below.

\begin{definition}
\label{def:SPP}
Let $X, Y$ be subshifts and let $f : X \to Y$ be a morphism. Define
\[ \mathcal{P}_p(Y) = \{ u \in \lang(Y) \;\vl\; u^\ZZ \in Y, \vl u\vl \leq p \}. \]
We say $f$ satisfies the \emph{strong $p$-periodic point condition} if there exists a length-preserving function $G : \mathcal{P}_p(Y) \to \lang(X)$ such that for all $u, v \in \mathcal{P}_p(Y)$ and $w \in \lang(Y)$ with $\INF u .w v \INF \in Y$, there exists an $f$-preimage for $\INF u .w v \INF$ of the form
$\INF G(u) w' . w'' w''' G(v) \INF \in X$
where $\vl u\vl$ divides $\vl w'\vl $, $\vl v\vl $ divides $\vl w'''\vl $ and $\vl w\vl  = \vl w''\vl $. The \emph{strong periodic point condition} is that the strong $p$-periodic point condition holds for all $p \in \N$.
\end{definition}

The condition states that we can pick, for each eventually periodic point, a preimage whose tails have the same eventual periods as the image, and that we can make these choices consistently (determined by a function $G$).

The strong periodic point condition is an obvious necessary condition for having a right inverse, as the right inverse must consistently pick preimages for periodic points, and since it is a CA, eventually it only sees the periodic pattern when determining the preimage, and begins writing a periodic preimage. Let us show the XOR CA with neighborhood $\{0,1\}$ is not regular using this method -- this is clear from the fact it is surjective and noninjective, or from the fact there are $1$-periodic points with no preimage of period $1$, but it also neatly illustrates the strong periodic point method.

\begin{example}
\label{ex:XOR}
The CA $f : \{0,1\}^\ZZ \to \{0,1\}^\ZZ$ defined by
\[ f(x)_i = 1 \iff x_i + x_{i+1} \equiv 1 \bmod 2 \]
is not regular. To see this, consider the strong $p$-periodic point condition for $p = 1$. Since $f(0^\ZZ) = f(1^\ZZ) = 0^\ZZ$, the point $0^\ZZ$ has two preimages, and we must have either $G(0) = 0$ or $G(0) = 1$. It is enough to show that neither choice of $a = G(0)$ 
is consistent, i.e. there is a point $y$ which is in the image of $f$ such that $y$ has no preimage that is left and right asymptotic to $a^\ZZ$. This is shown by considering the point
\[ y = ...0000001000000... \]
(which is in the image of $f$ since $f$ is surjective). It has two preimages, and the one left-asymptotic to $a^\ZZ$ is right-asymptotic to $(1-a)^\ZZ$. \qee 
\end{example}

In \cite[Theorem~1]{SaTo15a}, it is shown that the strong periodic point condition actually characterizes split epicness, in the case when $X$ is an SFT and $Y$ is a sofic shift.

\begin{theorem}
\label{thm:SplitEpicSolved}
Given two sofic shifts $X \subset S^\ZZ$ and $Y \subset R^\ZZ$ and a morphism $f : X \to Y$, it is decidable whether $f$ is split epic. If $X$ is an SFT, split epicness is equivalent to the strong periodic point condition.
\end{theorem}

We note that Definition~\ref{def:SPP} is equivalent to a variant of it where $G$ is only defined on Lyndon words \cite{Lo02}, i.e. lexicographically minimal representative words of periodic orbits: if $G$ is defined on those, it can be extended to all of $\mathcal{P}_p$ in an obvious way, and the condition being satisfied by minimal representatives implies it for all eventually periodic points.

\begin{remark}
It is observed in \cite[Theorem~1]{CaGa20} that if $f : X \to Y$ is split epic, then every periodic point in $Y$ must have a preimage of the same period in $X$ -- this is a special case of the above, and could thus be called the \emph{weak periodic point condition}. In \cite[Example~5]{SaTo15a}, an example is given of morphism between mixing SFTs which satisfies the weak periodic point condition but not the strong one. We will see later that ECA 9 and ECA 28 are non-regular, but satisfy the weak periodic point condition. In \cite[Theorem~4]{CaGa20}, for full shifts on finite groups, the weak periodic point condition is shown to be equivalent to split epicness (when CA are considered to be morphisms onto their image). In the context of CA on $\ZZ^2$, there is no useful strong periodic point condition in the sense that split epicness is undecidable, see Corollary~\ref{cor:2D}. 
\end{remark}

In the proof of Theorem~\ref{thm:SplitEpicSolved} in \cite{SaTo15a}, decidability is obtained from giving a bound on the radius of a minimal inverse, and a very large one is given, as we were only interested in the theoretical decidability result. The method is, however, quite reasonable in practise:
\begin{itemize}
\item To semidecide non-(split epicness), look at periodic points one by one, and try out different possible choices for their preimages. Check by automata-theoretic methods (or ``by inspection'') which of these are consistent in the sense of Definition~\ref{def:SPP}.
\item To semidecide split epicness, invent a right inverse -- note that here we can use the other semialgorithm (running in parallel) as a tool, as it tells us more and more information about how the right inverse must behave on periodic points, which tells us more and more values of the local rule.
\end{itemize}
One of these is guaranteed to finish eventually by \cite{SaTo15a}. For the regular elementary cellular automata, a simpler method sufficed, we essentially used only the weak periodic point condition combined with brute force search, see Section~\ref{sec:Searches}.

We finish this section with two more conditions for regularity. Proposition~\ref{prop:soficnotsplit} below is a slight generalization of \cite[Proposition~1]{SaTo15a}. We give a proof here, as the proof in \cite{SaTo15a} unnecessarily applies a more difficult result of S. Taati (and thus needs the additional assumption of ``mixing''). 

\begin{lemma}
If $X$ is an SFT and $f : X \to X$ is idempotent, i.e. $f^2 = f$, then $f(X)$ is an SFT.
\end{lemma}

\begin{proof}
Clearly $x \in f(X) \iff f(x) = x$, which is an SFT condition. Namely, if the radius of $f$ is $r$, the forbidden patterns are the words of length $2r+1$ where the local rule changes the value of the current cell. 
\end{proof}

\begin{proposition}
\label{prop:soficnotsplit}
If $X$ is an SFT and $f : X \to X$ is regular, then $f(X)$ is of finite type.
\end{proposition}

\begin{proof}
Let $g : X \to X$ be a weak inverse. Then $g \circ f : X \to X$ is idempotent, so $g(f(X))$ is an SFT. Note that the domain-codomain restriction $g\vl_{f(X), g(f(X))} : f(X) \to g(f(X))$ is a \emph{conjugacy}, meaning shift-commuting homeomorphism, between $f(X)$ and $g(f(X))$: its two-sided inverse is $f\vl_{g(f(X)), f(X)} : g(f(X)) \to f(X)$ by a direct computation. Thus $f(X)$ is also an SFT, since being an SFT is preserved under conjugacy \cite[Theorem~2.1.10]{LiMa95}.
\end{proof}

We also mention another condition, Proposition~\ref{prop:SurjectiveNotInjective} below, although it is not applicable to the ECA we consider.

\begin{lemma}
\label{lem:InjImpliesSurj}
Let $X$ be a subshift with dense periodic points and $f : X \to X$ a CA. If $f$ is injective, it is surjective.
\end{lemma}

\begin{proof}
The set $X_p = \{x \in X \;\vl\; \sigma^p(x) = x\}$ satisfies $f(X_p) \subset X_p$. Since $f$ is injective and $X_p$ is finite, we must have $f(X_p) = X_p$. Thus $f(X)$ is a closed set containing the periodic points. If periodic points are dense, $f(X) = X$. 
\end{proof}

We are interested mainly in mixing SFTs, where periodic points are easily seen to be dense. We remark in passing that in the case of mixing SFTs, the previous lemma can also be proved with an entropy argument: An injective CA cannot have a \emph{diamond}\footnote{This means a pair of distinct words whose long prefixes and suffixes agree, and which the local rule maps the same way, see \cite{LiMa95}.} when seen as a block map, so \cite[Theorem~8.1.16]{LiMa95} shows that the entropy of the image $f(X)$ of an injective CA is equal to the entropy of $X$. By \cite[Corollary~4.4.9]{LiMa95}, a mixing SFT $X$ is \emph{entropy minimal}, that is, it has no proper subshifts of the same entropy, and it follows that $f(X) = X$.

\begin{proposition}
\label{prop:SurjectiveNotInjective}
Let $X$ be a mixing SFT and $f : X \to X$ a surjective CA. Then $f$ is injective if and only if it is regular.
\end{proposition}

\begin{proof}
Suppose $f$ is a surjective CA on a mixing SFT. If it is also injective, it is bijective, thus reversible (by compactness of $X$), thus regular. Conversely, let $f$ be surjective and regular, and let $g : X \to X$ be a weak inverse. Then $g$ is injective, so it is surjective by the previous lemma. Thus $f$ must be bijective as well. 
\end{proof}

More generally, the previous proposition works on \emph{surjunctive subshifts} in the sense of \cite[Exercise~3.29]{CeCo10}, i.e.\ subshifts (on groups) where injective cellular automata are surjective. In particular this is the case for full shifts on surjunctive groups \cite{Go73,We00} such as abelian ones. Since injectivity is undecidable for surjective CA on $\ZZ^d$, $d \geq 2$ by \cite{Ka94a}, we obtain the following corollary.

\begin{corollary}
\label{cor:2D}
Regularity is undecidable for two-dimensional cellular automata. In fact, given a surjective CA $f : \Sigma^{\ZZ^2} \to \Sigma^{\ZZ^2}$, it is undecidable whether $f$ is split epic.
\end{corollary}

\section{Von Neumann non-regularity of elementary CA}
\label{sec:NonRegularityOfECA}

In \cite{CaGa20}, regularity was resolved for all ECA where the weak periodic point condition failed for period at most three, or there was a weak inverse directly among the ECA. The remaining cases up to equivalence are $6,7,9,23,27,28,33,41,57,58,77$. In this section, we prove the non-regular cases.

\begin{theorem}
The elementary CA with numbers 9, 27, 28, 41 and 58 are not regular.
\end{theorem}

\begin{proof}
See the lemmas below. 
\end{proof}

For the reader's convenience, we include the de Bruijn graphs in Figure~\ref{fig:DeBruijnGraphs}.

\begin{figure}
\centering
 \begin{subfigure}[b]{0.3\textwidth}
         \centering
        \begin{tikzpicture}[scale=1.8]
\node[draw,circle,inner sep=2] (zz) at (0,0) {$00$};
\node[draw,circle,inner sep=2] (zo) at (1,0) {$01$};
\node[draw,circle,inner sep=2] (oz) at (0,-1) {$10$};
\node[draw,circle,inner sep=2] (oo) at (1,-1) {$11$};
\draw (zz) edge[loop left] node {$1$} (zz);
\draw (oo) edge[loop right] node {$0$} (oo);
\draw (zz) edge[->] node[above] {$0$} (zo);
\draw (oo) edge[->] node[below] {$0$} (oz);
\draw (oz) edge[->] node[left] {$0$} (zz);
\draw (zo) edge[->] node[right] {$1$} (oo);
\draw (zo) edge[->,bend left=10] node[below right] {$0$} (oz);
\draw (oz) edge[->,bend left=10] node[above left] {$0$} (zo);
\end{tikzpicture}
         \caption{ECA 9.}
         \label{fig:ECA9}
     \end{subfigure}
     \hfill
 \begin{subfigure}[b]{0.3\textwidth}
         \centering
        \begin{tikzpicture}[scale=1.8]
\node[draw,circle,inner sep=2] (zz) at (0,0) {$00$};
\node[draw,circle,inner sep=2] (zo) at (1,0) {$01$};
\node[draw,circle,inner sep=2] (oz) at (0,-1) {$10$};
\node[draw,circle,inner sep=2] (oo) at (1,-1) {$11$};
\draw (zz) edge[loop left] node {$1$} (zz);
\draw (oo) edge[loop right] node {$0$} (oo);
\draw (zz) edge[->] node[above] {$1$} (zo);
\draw (oo) edge[->] node[below] {$0$} (oz);
\draw (oz) edge[->] node[left] {$1$} (zz);
\draw (zo) edge[->] node[right] {$1$} (oo);
\draw (zo) edge[->,bend left=10] node[below right] {$0$} (oz);
\draw (oz) edge[->,bend left=10] node[above left] {$0$} (zo);
\end{tikzpicture}
         \caption{ECA 27.} 
         \label{fig:ECA27}
     \end{subfigure}
     \hfill
 \begin{subfigure}[b]{0.3\textwidth}
         \centering
        \begin{tikzpicture}[scale=1.8]
\node[draw,circle,inner sep=2] (zz) at (0,0) {$00$};
\node[draw,circle,inner sep=2] (zo) at (1,0) {$01$};
\node[draw,circle,inner sep=2] (oz) at (0,-1) {$10$};
\node[draw,circle,inner sep=2] (oo) at (1,-1) {$11$};
\draw (zz) edge[loop left] node {$0$} (zz);
\draw (oo) edge[loop right] node {$0$} (oo);
\draw (zz) edge[->] node[above] {$0$} (zo);
\draw (oo) edge[->] node[below] {$0$} (oz);
\draw (oz) edge[->] node[left] {$1$} (zz);
\draw (zo) edge[->] node[right] {$1$} (oo);
\draw (zo) edge[->,bend left=10] node[below right] {$1$} (oz);
\draw (oz) edge[->,bend left=10] node[above left] {$0$} (zo);
\end{tikzpicture}
         \caption{ECA 28.} 
         \label{fig:ECA28}
     \end{subfigure}
     \newline
\hspace*{1.5cm}
 \begin{subfigure}[b]{0.3\textwidth}
         \centering
        \begin{tikzpicture}[scale=1.8]
\node[draw,circle,inner sep=2] (zz) at (0,0) {$00$};
\node[draw,circle,inner sep=2] (zo) at (1,0) {$01$};
\node[draw,circle,inner sep=2] (oz) at (0,-1) {$10$};
\node[draw,circle,inner sep=2] (oo) at (1,-1) {$11$};
\draw (zz) edge[loop left] node {$1$} (zz);
\draw (oo) edge[loop right] node {$0$} (oo);
\draw (zz) edge[->] node[above] {$0$} (zo);
\draw (oo) edge[->] node[below] {$0$} (oz);
\draw (oz) edge[->] node[left] {$0$} (zz);
\draw (zo) edge[->] node[right] {$1$} (oo);
\draw (zo) edge[->,bend left=10] node[below right] {$0$} (oz);
\draw (oz) edge[->,bend left=10] node[above left] {$1$} (zo);
\end{tikzpicture}
         \caption{ECA 41.} 
         \label{fig:ECA41}
     \end{subfigure}
     \hfill
 \begin{subfigure}[b]{0.3\textwidth}
 \centering
        \begin{tikzpicture}[scale=1.8]
\node[draw,circle,inner sep=2] (zz) at (0,0) {$00$};
\node[draw,circle,inner sep=2] (zo) at (1,0) {$01$};
\node[draw,circle,inner sep=2] (oz) at (0,-1) {$10$};
\node[draw,circle,inner sep=2] (oo) at (1,-1) {$11$};
\draw (zz) edge[loop left] node {$0$} (zz);
\draw (oo) edge[loop right] node {$0$} (oo);
\draw (zz) edge[->] node[above] {$1$} (zo);
\draw (oo) edge[->] node[below] {$0$} (oz);
\draw (oz) edge[->] node[left] {$1$} (zz);
\draw (zo) edge[->] node[right] {$1$} (oo);
\draw (zo) edge[->,bend left=10] node[below right] {$0$} (oz);
\draw (oz) edge[->,bend left=10] node[above left] {$1$} (zo);
\end{tikzpicture}
         \caption{ECA 58.}
         \label{fig:ECA58} 
     \end{subfigure}
\hspace*{1.5cm}
\caption{De Bruijn graphs of the non-regular ECA.}
\label{fig:DeBruijnGraphs}
\end{figure}
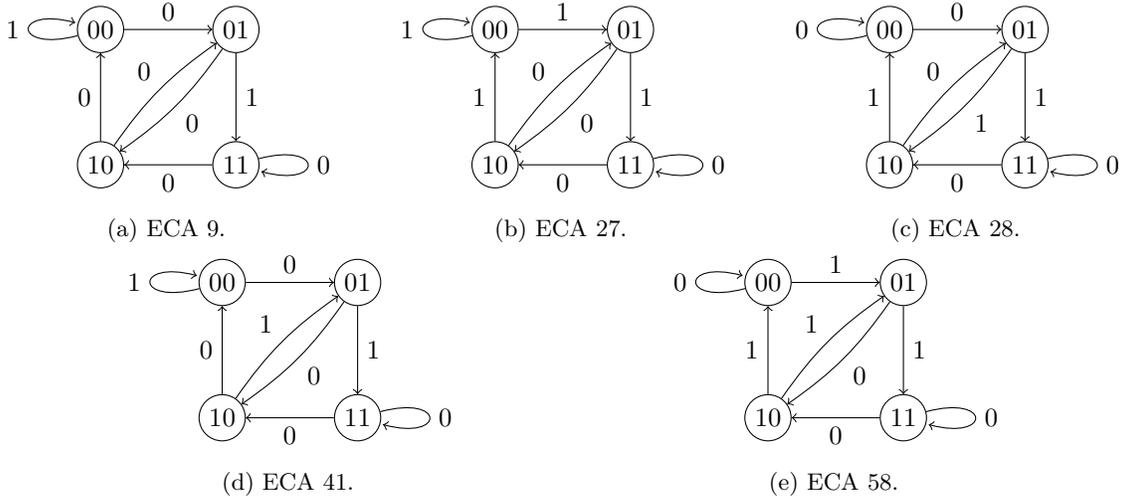

\begin{lemma}
\label{lem:9}
The elementary CA 9 is not regular.
\end{lemma}

\begin{proof}
Let $f$ be the ECA $9$, i.e.\ $f(x)_i = 1 \iff x_{[i-1,i+1]} \in \{000, 011\}$. Let $X$ be the image of $f$.


We have $f(0^\ZZ) = 1^\ZZ$ and $f(1^\ZZ) = 0^\ZZ$, so if $g : X \to \{0,1\}^\ZZ$ is a right inverse for $f$, then $g(0^\ZZ) = 1^\ZZ$. Consider now the configuration
\[ x = ...000011.00000... \]
where coordinate $0$ is to the left of the decimal point (i.e. at the rightmost $1$ of the word $11$). 

We have $x \in X$ because $f(...10101000.010101...) = x$. (Alternatively, by Proposition~\ref{prop:ECA9Image}, $X$ is precisely the SFT defined by forbidden words $1011$, $10101$, $11001$, $11000011$ and $110000101$.)

Let $g(x) = y$. Then $y_i = 1$ for all large enough $i$ and $y_i = 0$ for some $i$. Let $n$ be maximal such that $y_n = 0$. Then $y_{[n,n+2]} = 011$ so $f(y)_{n+1} = 1$ and $f(y)_{n+1+i} = 0$ for all $i \geq 1$. Since $f(y) = x$, we must have $n = -1$ and since $\{000, 011\}$ does not contain a word of the form $a01$, it follows that $f(y)_{-1} = 0 \neq x_{-1}$, a contradiction. 
\end{proof}

The proof above shows that the ECA $9$ does not have the strong periodic point property for $p = 1$. In general, for fixed $p$ one can use automata theory to decide whether it holds up to that $p$, though here (and in all other proofs) we found the contradictions by hand before we had to worry about implementing this.

Next, we cover ECA $27$. The image $X$ of $f$ can be shown to be proper sofic, so Proposition~\ref{prop:soficnotsplit} directly shows that the CA can not be regular. Nevertheless, we give a direct proof to illustrate the method.

\begin{lemma}
The ECA $27$ is not regular.
\end{lemma}


\begin{proof}
Let $f$ be the ECA $27$, i.e.\ $f(x)_i = 1 \iff x_{[i-1,i+1]} \in \{000, 001, 011, 100\}$. Let $X$ be the image of $f$.

Again, we will see that this CA does not satisfy the strong periodic point condition for $p = 1$. Observe that $f(1^\ZZ) = 0^\ZZ$ and $f(0^\ZZ) = 1^\ZZ$ so if $g$ is a right inverse from the image to $\{0,1\}^\ZZ$, then $g(0^\ZZ) = 1^\ZZ$ and $g(1^\ZZ) = 0^\ZZ$.

We have
\begin{align*}
f(...&000001100.10101010...) = \\
         ...&111111011.00000000... = x \in X.
\end{align*}
We now reason similarly as in Lemma~\ref{lem:9}. We have $g(x)_i = 1$ for all large enough $i$, and if $n$ is maximal such that $g(x)_n = 0$, then $f(g(x))_{n+1} = 1$ and $f(g(x))_{n+1+i} = 0$ for all $i \geq 1$, so again necessarily $n = -1$. A short combinatorial analysis shows that no continuation to the left from $n$ produces $f(g(x))_{n} = 1$ and $f(g(x))_{n-1} = 0$, that is, the image of $g$ has no possible continuation up to coordinate $-3$. 
\end{proof}

\begin{lemma}
The ECA $28$ is not regular.
\end{lemma}


\begin{proof}
Let $f$ be the ECA $28$, i.e.\ $f(x)_i = 1 \iff x_{[i-1,i+1]} \in \{010, 011, 100\}$. Let $X$ be the image of $f$.


We have $f(0^\ZZ) = f(1^\ZZ) = 0^\ZZ$. We have
\begin{align*}
f(...&1111.0000...) = \\
 ...&0000.1000... = x \in X.
\end{align*}
(Alternatively, $x \in X$ by Proposition~\ref{prop:ECA28Image} which states that $X$ is the SFT with the single forbidden word $111$.)

This contradicts the choice $g(0^\ZZ) = 0^\ZZ$ by a similar analysis as in Example~\ref{ex:XOR}: computing the preimage from right to left, the asymptotic type necessarily changes to $1$s. Thus we must have $g(0^\ZZ) = 1^\ZZ$.

On the other hand, if $g(0^\ZZ) = 1^\ZZ$, then going from right to left, we cannot find a consistent preimage for
\begin{align*} f(...&0001.00000...) = \\
 ...&0001.10000... \end{align*}
(Alternatively, going from left to right, the asymptotic type necessarily changes to $0$s or never becomes $1$-periodic.)

It follows that $g(0^\ZZ)$ has no consistent possible choice, a contradiction. 
\end{proof}

Next we consider ECA $41$. Again Proposition~\ref{prop:soficnotsplit} would yield the result, because the image is proper sofic. For ECA 41, in fact even the weak periodic point condition fails, see Section~\ref{sec:WPP}.

\begin{lemma}
The ECA $41$ is not regular.
\end{lemma}


\begin{proof}
Let $f$ be the ECA $41$, i.e.\ $f(x)_i = 1 \iff x_{[i-1,i+1]} \in \{000, 011, 101\}$. 

We have $f(1^\ZZ) = 0^\ZZ$ and $f(0^\ZZ) = 1^\ZZ$. In the usual way (right to left), we verify that the point 
\begin{align*}
f(...&00000001.00100100...) = \\
         ...&11111100.00000000...  \in X.
\end{align*}
has no preimage that is right asymptotic to $1^\ZZ$, obtaining a contradiction. 
\end{proof}

Next we consider ECA $58$. Again Proposition~\ref{prop:soficnotsplit} would also yield the result, since the image is proper sofic.

\begin{lemma}
The ECA $58$ is not regular.
\end{lemma}


\begin{proof}
Let $f$ be the ECA $58$, i.e.\ $f(x)_i = 1 \iff x_{[i-1,i+1]} \in \{001, 011, 100, 101\}$.

The point $0^\ZZ$ has two $1$-periodic preimages. We show neither choice satisfies the strong periodic point condition: if $g(0^\ZZ) = 1^\ZZ$, then $g$ cannot give a preimage for
\begin{align*}
f(...&11111.0000...) = \\
          ...&00000.1000...
\end{align*}
If $g(0^\ZZ) = 0^\ZZ$, then it cannot give a preimage for
\begin{align*}
f(...&00000.01111...) = \\
          ...&00000.11000...
\end{align*}
Thus, ECA 58 is not regular.
\end{proof}

%

\section{Weak periodic point condition and SFT images} 
\label{sec:WPP}

In this section we look at the weak periodic point condition (WPP) and the condition of having SFT images, for the five non-regular CA that were left open in \cite{CaGa20} (where the WPP was checked up to period $3$) and for which we showed the strong periodic point condition (SPP) fails in the previous section. The results are summarized in Table~\ref{tab:Conditions}, ``no'' means that the particular method can be used to conclude non-regularity (because the necessarily condition does not hold), ``yes'' means that it can not.

\begin{table}
\begin{center}
\begin{tabular}{|l|l|l|l|}
\hline
         & SFT image & WPP & SPP \\
\hline
ECA $9$  & yes       & yes & no \\
\hline
ECA $27$ & no        & yes & no \\
\hline
ECA $28$ & yes       & yes & no \\ 
\hline
ECA $41$ & no        & no  & no \\
\hline
ECA $58$ & no        & yes & no \\
\hline
\end{tabular}
\end{center}
\caption{Necessary conditions for regularity satisfied by the five non-regular ECA.}
\label{tab:Conditions}
\end{table}

\subsection{Weak periodic point condition}

We show that the ECA $9$, $27$, $28$ and $58$ satisfy the weak periodic point condition, i.e.\ every periodic point in the image has a preimage with the same period, while $41$ does not.

\begin{proposition}
ECA 41 does not satisfy the weak periodic point condition.
\end{proposition}

\begin{proof}
In the local rule of this ECA $f$, we have $000, 011, 101 \rightarrow 1$, others map to $0$. We have $f((000101)^\ZZ) = (010)^\ZZ$, $f((001)^\ZZ) = 0^\ZZ, f((011)^\ZZ) = (110)^\ZZ$, so $(010)^\ZZ$ has a $6$-periodic preimage, but no $3$-periodic preimage.
\end{proof}

\begin{proposition}
\label{prop:ECA28WPP}
ECA 28 satisfies the weak periodic point condition.
\end{proposition}

\begin{proof}
In this ECA, $010, 011, 100 \rightarrow 1$, others map to $0$. It is easy to check that this ECA amounts to a right shift followed by a bit flip when restricted to the flipped golden mean shift, i.e.\ the SFT with the unique forbidden word $00$. Thus, on the golden mean shift (the SFT with unique forbidden word $11$), a section for this CA is left shift composed with bit flip.

Thus, any periodic point in the golden mean shift will have a preimage with the same period. Consider then a $p$-periodic point $x$ not in the golden mean shift, we may assume $x = (11u)^\ZZ$ for some $u$, say $\vl u\vl  = p-2$. Now take any preimage $y$ for $x$, and observe that necessarily $y_{[n-1,n+3]} = 0100$ for all $n \in p\ZZ$, by a quick look at the local rule. In automata theoretic terms, $11$ is a synchronizing word for the de Bruijn automaton. This synchronization means that also $(y_{[0,p-1]})^\ZZ$ is a preimage.
\end{proof}

\begin{proposition}
ECA 9 satisfies the weak periodic point condition.
\end{proposition}

\begin{proof}
In this one, $000, 011 \rightarrow 1$, others map to $0$. It is easy to check that on the subshift $\lang^{-1}((1000^*)^*)$ (on the image side), a section is again given by left shift composed with bit flip. The words $11$, $101$ are easily seen to force a unique preimage word of length at least $2$ (i.e.\ are synchronizing for the automaton), and we can argue as in the previous proof for all periodic points containing one of these words.
\end{proof}

For ECA $27$ and $58$ we decided the weak periodic point condition by computer.

\begin{proposition}
\label{prop:SFTDecidable}
The weak periodic point condition is decidable.
\end{proposition}

\begin{proof}
Let $f$ be a CA. The language of words $(u,v)$ of the same length such that $f(u^\ZZ) = v^\ZZ$ is clearly regular, let $L$ be its right projection. The image of $f$ is sofic, so let $K$ be the language of words $v$ such that $v^\ZZ$ is in the image subshift of $f$, which is also clearly regular. If every periodic point admits an $f$-preimage of the same period, then clearly $K = L$. If $u^\ZZ$ does not admit a point of the same period, and $\vl u\vl $ is minimal for this point, then $u \in K \setminus L$.
\end{proof}

An implementation (as a Python/SAGE script) is included in \cite{Sa22RegularScripts} as {\tt check\_WPP.sage}.

\begin{proposition}
ECA 27 and 58 satisfy the weak periodic point condition.
\end{proposition}

\subsection{SFT images}
\label{sec:SFTImages}

We show that ECA $9$ and $28$ have SFT images, while the others do not. We begin with the observation that one can do this easily by computer:

\begin{proposition}
\label{prop:SFTDecidable}
It is decidable whether the image of a CA is an SFT.
\end{proposition}

\begin{proof}
The image subshift is sofic. To check if a sofic shift is an SFT, observe that by \cite[Exercise~1.3.8]{LiMa95}, a subshift is an SFT if and only if the set of \emph{first offenders}, i.e.\ the minimal forbidden words, is finite. If $L$ is the complement of the language of $X \subset \Sigma^\ZZ$, then the first offenders are by definition exactly the language $L \setminus (\Sigma^+ L \cup L \Sigma^+)$, where $\Sigma^+$ denotes the non-empty words over $\Sigma$; regular languages are effectively closed under these operations and it is trivial to check if a regular language is finite.
\end{proof}

In the remainder of this section, we give more or less ad hoc human proofs of SFTness and proper soficness of images, not following this algorithm.

\begin{proposition}
\label{prop:ECA28Image}
ECA 28 has SFT image, defined by a single forbidden word $111$.
\end{proposition}

\begin{proof}
Again recall that this ECA is defined by
\[ f(x)_i = 1 \iff x_{[i-1,i+1]} \in \{010, 011, 100\}. \]
The reader may also find the labeled de Bruijn graph given in Figure~\ref{fig:DeBruijnGraphs} helpful.

As we saw in the proof of Proposition~\ref{prop:ECA28WPP}, the only preimage of $11$ is $0100$, so $111$ does not have a preimage.

We describe a procedure (in fact, a type of transducer) giving preimages for transitive points (points containing every finite word that appears in the image, which are clearly dense). As we saw in the proof of Proposition~\ref{prop:ECA28WPP}, the only preimage of $11$ is $0100$, so $111$ does not have a preimage. Now, consider a transitive point in the SFT with $111$ forbidden (note that this SFT is mixing, so one exists). To construct a preimage, on top of every $11$ put the word $0100$. Now start filling the gaps to the right from each $0100$ already filled. For the leftmost run of $0$s, write $0$s on top. When you reach the first $1$, write $1$ on top of it. If it is part of a $11$, you have completed the run successfully. Otherwise, start filling according to the shift-and-flip-rule.
\end{proof}


\begin{proposition}
ECA $27$, $41$ and $58$ have proper sofic images.
\end{proposition}

\begin{proof}
For ECA $27$,
the word $0100$ is synchronizing for the de Bruijn automaton. It is then easy to verify that for large $n$, $\INF 0 1 1 0 (00)^n 1 \INF$ is in the image subshift, but $\INF 0 1 1 0 0 (00)^n 1 \INF$ is not, which immediately shows that the image is proper sofic.

For ECA $41$, we similarly see that $111$ is synchronizing, and then that $\INF 0 111 00 (000)^n 11 0 \INF$ is in the image subshift, but $\INF 0 111 0 (000)^n 11 0 \INF$ is not. For ECA $58$ we similarly see that $0100$ is synchronizing, and from this it is clear that $\INF 0 1 0 \INF$ is in the image subshift, but $\INF 0 1 0^n 1 0 \INF$ is not.
\end{proof}

\begin{remark}
We found the above proofs by performing the subset construction on the de Bruijn automaton and looking for cycles. Analyzing cycles is a general way to prove proper soficity, in the sense that if we have a DFA for the language of a proper sofic subshift, we can always find words $u, v, w$ such that $uv^n$ and $v^nw$ are in the language for all $n$, but the words $u v^n w$ are not (for any $n$). To see this, apply a suitable variant of the pumping lemma to a long first offender. Such a triple is also easy to verify from the DFA. In the case of the ECA, we found cycles consisting of singleton states (in the subset DFA), and following the suggestion of the anonymous referee, we rephrased the proof in terms of synchronizing words instead.
\end{remark}

\begin{proposition}
\label{prop:ECA9Image}
ECA 9 has SFT image with first offenders
\[1011, 10101, 11001, 11000011, 110000101. \]
\end{proposition}

\begin{proof}
Performing the subset construction and complementation on the de Bruijn automaton, and merging the state $\{00,01,10\}$ into $\{0,1\}^2$ (as they are clearly equivalent), we obtain a DFA for the orphans of ECA $9$.

The DFA is shown in Figure~\ref{fig:DFAECA9}. By renaming the states by the maximal relevant suffix we have seen, we can directly interpret this as a DFA for one-way infinite words that lack the listed offenders, see Figure~\ref{fig:Renamed}. The main thing to note is that having seen the $1100001$ is equivalent to having seen $101$, as continuations $01, 1$ fail and continuation $00$ resets the state.
\end{proof}

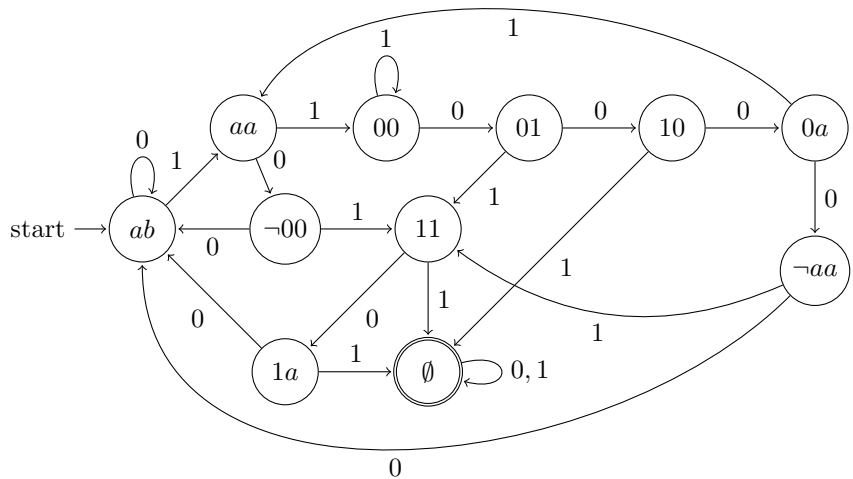
\begin{figure}
\begin{center}
\begin{tikzpicture}[shorten >=1pt,node distance=1.9cm,on grid,auto]
\node[state,initial]	(all)						{$ab$};
\node[state]			(zzoo)	[above right=of all]	{$aa$};
\node[state]			(zoozoo) [right=of all]	{$\neg 00$};
\node[state]			(zz)		[right=of zzoo]		{$00$};
\node[state]			(zo)		[right=of zz]			{$01$};
\node[state]			(oz)		[right=of zo]			{$10$};
\node[state]			(zzzo)	[right=of oz]		{$0a$};
\node[state]			(zooz)	[below=of zzzo]		{$\neg aa$};
\node[state]			(oo)		[right=of zoozoo]		{$11$};
\node[state,accepting]	(empty)	[below=of oo]			{$\emptyset$};
\node[state]			(ozoo)	[left=of empty]		{$1a$};

\path[->]	(all) 	edge [loop above] node {$0$} (all)
				edge 			node {$1$} (zzoo)
		(zzoo)	edge 			node {$0$} (zoozoo)
				edge 			node {$1$} (zz)
		(zoozoo)	edge 			node {$0$} (all)
				edge 			node {$1$} (oo)
		(zz)		edge 			node {$0$} (zo)
				edge [loop above]	node {$1$} (zz)
		(zo)		edge 			node {$0$} (oz)
				edge 			node {$1$} (oo)
		(oz)		edge 			node {$0$} (zzzo)
				edge 			node {$1$} (empty)
		(zzzo)	edge 			node {$0$} (zooz)
				edge [out=135,in=60,looseness=0.75]	node {$1$} (zzoo)
		(zooz)	edge [out=225,in=270, looseness=1.05] 			node {$0$} (all)
				edge [bend left=30]		node {$1$} (oo)
		(oo)		edge 			node {$0$} (ozoo)
				edge 			node {$1$} (empty)
		(ozoo)	edge 			node {$0$} (all)
				edge 			node {$1$} (empty)
		(empty)	edge [loop right]  node {$0,1$} (empty)
;
\end{tikzpicture}
\end{center}
\caption{DFA for orphans of ECA 9. For conciseness we use the binary variables $a,b$ to express arbitrary bits, e.g.\ $aa$ represents the state $\{00, 11\}$. By $\neg$ we denote complement, e.g. $\neg aa$ represents $\{01, 10\}$.}
\label{fig:DFAECA9}
\end{figure}

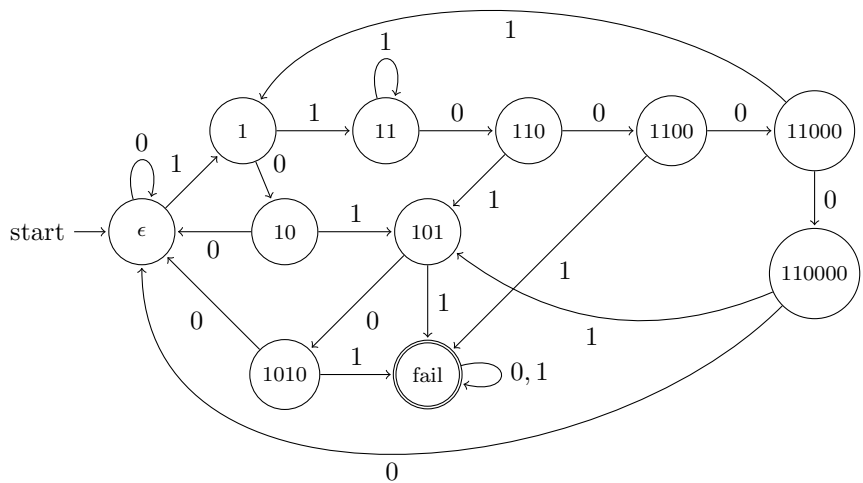
\begin{figure}
\begin{center}
\begin{tikzpicture}[shorten >=1pt,node distance=1.9cm,on grid,auto]
\node[state,initial]	(all)						{\footnotesize$\epsilon$};
\node[state]			(zzoo)	[above right=of all]	{\footnotesize$1$};
\node[state]			(zoozoo) [right=of all]		{\footnotesize$10$};
\node[state]			(zz)		[right=of zzoo]		{\footnotesize$11$};
\node[state]			(zo)		[right=of zz]			{\footnotesize$110$};
\node[state]			(oz)		[right=of zo]			{\footnotesize$1100$};
\node[state]			(zzzo)	[right=of oz]			{\footnotesize$11000$};
\node[state]			(zooz)	[below=of zzzo]		{\footnotesize$110000$};
\node[state]			(oo)		[right=of zoozoo]		{\footnotesize$101$};
\node[state,accepting]	(empty)	[below=of oo]			{\footnotesize fail};
\node[state]			(ozoo)	[left=of empty]		{\footnotesize$1010$};

\path[->]	(all) 	edge [loop above] node {$0$} (all)
				edge 			node {$1$} (zzoo)
		(zzoo)	edge 			node {$0$} (zoozoo)
				edge 			node {$1$} (zz)
		(zoozoo)	edge 			node {$0$} (all)
				edge 			node {$1$} (oo)
		(zz)		edge 			node {$0$} (zo)
				edge [loop above]	node {$1$} (zz)
		(zo)		edge 			node {$0$} (oz)
				edge 			node {$1$} (oo)
		(oz)		edge 			node {$0$} (zzzo)
				edge 			node {$1$} (empty)
		(zzzo)	edge 			node {$0$} (zooz)
				edge [out=135,in=60,looseness=0.75]	node {$1$} (zzoo)
		(zooz)	edge [out=225,in=270, looseness=1.05] 			node {$0$} (all)
				edge [bend left=30]		node {$1$} (oo)
		(oo)		edge 			node {$0$} (ozoo)
				edge 			node {$1$} (empty)
		(ozoo)	edge 			node {$0$} (all)
				edge 			node {$1$} (empty)
		(empty)	edge [loop right]  node {$0,1$} (empty)
		;
\end{tikzpicture}
\end{center}
\caption{The DFA from Figure~\ref{fig:DFAECA9} with states renamed with the suffix they are tracking.}
\label{fig:Renamed}
\end{figure}

\section{Von Neumann regularity of elementary CA}
\label{sec:Searches}

\begin{theorem}
The elementary CA with numbers 6, 7, 23, 33, 57 and 77 are regular.
\end{theorem}

\begin{proof}
We found weak inverses by computer. A program (written for Python version 3.9) for finding inverses can be found at \cite{Sa22RegularScripts} in {\tt find\_weak\_inverses.py}. In Section~\ref{sec:Procedure} we describe the simple procedure used to find the weak inverses, and in Section~\ref{sec:Inverses} we give hexadecimal codes for them.
\end{proof}

\subsection{The procedure}
\label{sec:Procedure}

To find the weak inverses for $6$, $7$, $23$, $33$, $57$ and $77$, we used the following very simple procedure (executed by computer). We go through the possible radii in increasing order, and for each we enumerate the (possibly empty) list of weak inverses of that radius, until we find the first radius for which an inverse exists.

For a given radius $r$, the weak inverses $g$ can be listed as follows. First, we consider all periodic points up to some period $p$, and compute their periodic images under the ECA $f$ being considered. Each periodic point in the image of $f$ must be mapped by $g$ to one of its $f$-preimages with the same period, and we keep track of the possible choices with a mapping $d$ from periodic points (i.e.\ cyclic words) to their possible $g$-images. We then compute the words of length $2r+1$ in the image of the ECA, and construct a partial local rule for the weak inverse, call this local rule $\ell$ (initially all the entries hold an unknown value $?$).

We then iterate the following simplification steps, which are simply the obvious mutual consistency conditions for $d$ and $\ell$, until reaching a fixed point:
\begin{itemize}
\item If a periodic point has only one preimage left in $d$, set the corresponding values in $\ell$ and drop this point from $d$.
\item If some possible choices in $d$ are inconsistent with $\ell$, remove them.
\end{itemize}
If after reaching the fixed point we have not yet determined the $\ell$-image of some word, pick a random such word, recursively try both values for the $\ell$-image and enumerate the results.

We initially also used the following consistency condition:
\begin{itemize}
\item If all the possible preimages of a particular periodic point have a particular bit $b$ in position $i$, then we can safely make $\ell$ map the word at $[i-r, i+r]$ to $b$. 
\end{itemize}
We also experimented with more sensible choices for the choice of a new image in $\ell$. With a naive implementation, these only made the calculation slower, but such ideas (and even the strong periodic point condition) might be more useful with larger neighborhoods and state sets.

Finding the inverses using this procedure takes a few seconds with the highly non-optimized Python program {\tt find\_weak\_inverses.py} in \cite{Sa22RegularScripts}. For ECA $7$, finding the unique weak inverse takes only a few minutes when performed by hand. The optimal radii and $p$ used are listed in Table~\ref{tab:Optimals}, and we give the actual inverses in the following section.

\begin{table}
\begin{center}
\begin{tabular}{|l|l|l|l|}
 \hline 
 & optimal $r$ & number of inverses & recommended $p$ \\
 \hline
ECA 6 & 3 & 2 & 9 \\
\hline
ECA 7 & 2 & 1 & 6 \\
\hline
ECA 23 & 2 & 1 & 5 \\
\hline
ECA 33 & 2 & 1 & 6 \\
\hline
ECA 57 & 4 & 32 & 11 \\
\hline
ECA 77 & 2 & 1 & 6 \\
\hline
\end{tabular}
\end{center}
\caption{For the regular ECA, the optimal radius for a weak inverse, the number of inverses (counting only the possible behaviors on the words in the image of the ECA), and the $p$-value we used when executing the algorithm.}
\label{tab:Optimals}
\end{table}

\subsection{The weak inverses}
\label{sec:Inverses}

In the conference version \cite{Sa21a} of this paper, we picked the behavior of the weak inverse rules on an ad hoc basis on the words not in the image subshift, in order to get a nice presentation for each rule. Here, as we give the exhaustive list, it seems cleaner to simply output $0$ on the inputs that do not appear in the image subshift, which means that the rule numbers are different from those reported in the conference version even in the case when the weak inverse is ``unique''. 

To get all weak inverses of the optimal radius $r$, take one of the rules listed in the propositions below, compute the words $L$ in the image subshift of the ECA in question (for example by applying the ECA to all words of length $2r+3$), and change $0$ to $1$ on any of the words in the complement of $L$. These propositions are simply a summary of the output of the program {\tt find\_weak\_inverses.py} in \cite{Sa22RegularScripts}. There is also a separate (independent) script {\tt check\_weak\_inverses.py} in \cite{Sa22RegularScripts} that verifies these propositions.

\begin{proposition}
The radius $3$ weak inverses for ECA $6$, which map to $0$ outside the image subshift, have hex codes
\[ 00000e030f03000f00000e0f0f030e0{*} \]
where ${*}$ is replaced by one of $7, f$.
\end{proposition}

\begin{proposition}
The unique radius $2$ weak inverse for ECA $7$, which maps to $0$ outside the image subshift, has hex code $21232123$.
On the image subshift of ECA $7$, this is equivalent to the ECA $35$ composed with $\sigma$, which has hex code $23232323$.
\end{proposition}

\begin{proposition}
The unique radius $2$ weak inverse for ECA $23$, which maps to $0$ outside the image subshift, has hex code $23bb003b$.
\end{proposition}

\begin{proposition}
The unique radius $2$ weak inverse for ECA $33$, which maps to $0$ outside the image subshift, has hex code $00070707$.
\end{proposition}

\begin{proposition}
The radius $4$ weak inverses of ECA $57$, which map to $0$ outside the image subshift, are
\begin{align*} & 0000000000fffe0000e{*}feff0000feff0000fc{*}f0003feff00e{*}f{*}000000feff... \\
& 00000000000ffeff00e{*}fe0f0000feff0000fc{*}f00{*}3fe0000e{*}f{*}000000feff 
\end{align*}
(this is one word spread over two lines) where ${*}{*}{*}{*}{*}{*}{*}{*}{*}$ is replaced with one of
\[ 303e3003e, 307e3307e, bc3eb003e, bc7eb307e, 303630c36, 307633c76, bc36b0c36, \]
\[ bc76b3c76, 30b63ccb6, bcb6bccb6, 30be3c0be, bcbebc0be, 30f63fcf6, 30fe3f0fe, \]
\[ bcf6bfcf6, bcfebf0fe, 733e7003e, 737e7307e, ff3ef003e, ff7ef307e, 733670c36, \]
\[ ff36f0c36, 737673c76, ff76f3c76, 73be7c0be, 73b67ccb6, ffbefc0be, ffb6fccb6, \]
\[ 73f67fcf6, 73fe7f0fe, fffeff0fe, fff6ffcf6. \]
\end{proposition}

\begin{proposition}
The unique radius $2$ weak inverse for ECA $77$, which maps to $0$ outside the image subshift, has hex code $00733177$.
\end{proposition}

\section{Most cellular automata are non-regular}

In the conference version \cite{Sa21a} of this paper, we stated:
\begin{quote}[\textbf{...}] Based on this, one might conjecture that regularity is common in cellular automata as radius grows.
\end{quote}

We show that in fact most cellular automata are non-regular, for rather uninteresting reasons, namely the weak periodic point condition fails with high probability. Post-composition with a shift preserves regularity, so we restrict to CA with one-sided neighborhoods in the statement. By \emph{one-sided radius $r$} we refer to neighborhood $[0,r]$. Clearly all CA with one-sided radius $r = 0$ are regular, so we restrict to larger radii. Similarly, we do not consider alphabets with one or fewer letters.

\begin{theorem}
\label{thm:RandomCA}
Let $X_{n,r}$ be a uniformly random one-sided CA with radius $r$ over an alphabet of size $n$. Then as $n+r \rightarrow \infty$, with high probability $X_{n,r}$ is non-regular, and indeed the weak periodic point condition fails.
\end{theorem}

More precisely, for all $\epsilon > 0$ there exists $n_0$ such that for all $n+r \geq n_0$, $n \geq 2$, $r \geq 1$, $X_{n,r}$ is regular with probability at most $\epsilon$.

\begin{proof}
Let $\vl A \vl = n$ be the alphabet. Let us first show that for large $n$, $X_{n,r}$ is very likely to be non-regular no matter what the value of $r$ is. It turns out that it suffices to look at the behavior of points of period $1$ or $2$, namely we show that with high probability there is a point of period $1$ with a preimage of period $2$ but no preimage of period $1$. Let $P_1$ be the set of points of period $1$. For $a,b \in A$ with $a \neq b$ write $\iota(a,b) = abababa...$ and $\iota(a) = aaa..$. All that matters to us is how $f$ maps these points; note that this only depends on how the local rule maps inputs of the form $abab...$ for $a,b \in A$.

Note that the image of a point in $P_1$ is uniformly chosen out of points in $P_1$. Write $N$ for the random number of points in $P_1$ with a preimage in $P_1$, and write $I$ for this set. Next, let $M$ be the number of pairs $a, b \in A$, $a \neq b$, such that $f(\iota(a,b)) = f(\iota(b,a))$, meaning both $\iota(a,b)$ and $\iota(b,a)$ are mapped to a point in $P_1$.

It is clear that, conditioned on fixed values of $N$ and $M$, the following things are independent and uniformly distributed:
\begin{itemize}
\item the actual set $I$ (conditioned on $N$),
\item the $M$ many choices of images $f(\iota(a,b)) = f(\iota(b,a))$ for $a \neq b$,
\end{itemize}
From this, it follows that if we have $N < 0.9n$ and $M \geq k$, then the probability that some pair is mapped to a unary point that is not in $I$ is at least $1 - 0.9^k$, which tends to $1$ as $k \rightarrow \infty$. It remains to show that $N < 0.9n$ and $M \geq k$ with high probability where $k \rightarrow \infty$.

First, let us look at $N$. Writing $N_a \in \{0, 1\}$ for the random variable indicating that the symbol $a \in A$ has a preimage, $N = \sum_a N_a$, we have
\[ E(N) = \sum_a E(N_a) = n (1 - 1/n)^n = n(1 - 1/e) + O(1) \]
and
\begin{align*}
\Var(N) &= E(N^2) - E(N)^2 \\
&= \sum_a E(N_a^2) + \sum_{a \neq b} E(N_a N_b) - E(N)^2 \\
&= n(1 - 1/e)  +  n (n - 1) (1 - 2/e + 1/e^2)  -  n^2(1 - 1/e)^2   +  O(n) \\
&= O(n)
\end{align*}
Here recall that for fixed $k$, $(1 + k/n)^n = e^k + O(1/n)$.

The claim now follows from Chebyshev's inequality, which states that the difference between the value of a random variable and its mean is likely to be of the order of the square root of its variance. More precisely, writing $\Var(N) = \sigma^2 = \alpha n$ (where $\alpha = O(1)$ is technically a function of $n$) we have $0.2 n = 0.2 \sqrt{n/\alpha} \sigma$, and by Chebyshev's inequality
\begin{align*}
 P(N > 0.9n) &\leq P(\vl N-E(N) \vl > 0.2 n)  = \Prob(\vl N-E(N) \vl > 0.2 \sqrt{n/\alpha} \sigma) \\
& \leq 1 / 0.04 n \alpha = O(1 / n).
\end{align*}

Next, we look at $M$. Order the alphabet and write $M = \sum_{a < b} M_{a, b}$ where $M_{a, b} \in \{0, 1\}$ is the random variable where $M_{a,b} = 1$ denotes $f(\iota(a,b)) = f(\iota(b,a))$. We see that $M$ is simply a sum of independent random variables, thus follows the Bernoulli distribution $M \sim B(\binom{n}{2}, 1/n) = B(m, p)$. The expected value is $E(M) = pm = (n-1)/2$, and the variance is
\[ \Var(M) = mp(1-p) = \binom{n}{2} (1-1/n)/n = O(n), \]
so again with high probability the value is above, say $0.25 n$. Of course, we have $k = 0.25 n \rightarrow \infty$.

To conclude, it clearly suffices to show that for any fixed $n$, cellular automata are non-regular with high probability for large $r$. The argument is essentially the same as the above, with minor changes: Instead of unary points, we use points of period $p$ with $2p \leq r$ and $p$ large. Most words of length $p$ are \emph{primitive}, i.e.\ the corresponding periodic point has least period $p$, write again $P_p$ for this set. The set $P_p$ splits into equivalence classes under the shift, and if we pick the local rule at random, then the image of each $x \in P_p$ is picked uniformly from the set $Q_p$ of points of (not necessarily least) period $p$.

The images for different points in $P_p$ are independent from each other, so this is in fact the same distribution as we considered above, except that some periodic points may be mapped to a point of smaller period. Thus, by a coupling argument the probability that 90\% of points in $P_p$ have preimages in $P_p$ is clearly at most as large as it was in the previous process.

Now, we just have to show that there are at least $k$ points in $P_p$ with preimages in $P_{2p}$ with high probability, with $k$ tending to infinity with $r$. We observe that images of points in $P_{2p}$ are picked uniformly at random from points in $Q_{2p}$, and are independent from each other and from the choices of images of points in $P_p$. Thus, the calculation is the same as above, since given that a point in $P_{2p}$ maps to a point in $Q_p$, the probability that it did not map to $P_p$ is negligible.
\end{proof}


\section*{Acknowledgements}

We thank Jarkko Kari for observing that Proposition~\ref{prop:SurjectiveNotInjective} works in all dimensions. We thank Johan Kopra for pointing out that Lemma~\ref{lem:InjImpliesSurj} is easier to prove than to find in \cite{LiMa95}. We thank the anonymous referees for suggesting significant improvements to the presentation, and simpler proofs. In particular, the proof of Theorem~\ref{thm:RandomCA} was simplified essentially.

\bibliographystyle{vancouver}
\bibliography{../../../../bib/bib}{}


\end{document}